\documentclass[11pt]{article}

\usepackage{graphicx} 
\usepackage{algorithm}
\usepackage{algpseudocode}
\usepackage{svg}
\usepackage{enumitem}
\usepackage[margin=1in]{geometry}
\usepackage[square,numbers]{natbib}
\usepackage{booktabs}
\usepackage{amsmath}
\usepackage{amsthm}

\theoremstyle{plain}
\newtheorem{theorem} {Theorem} [section]

\theoremstyle{definition}

\newcommand{\lpecc}{\text{LP}_\text{ECC}}
\newcommand{\lpvc}{\text{LP}_\text{VC}}
\newcommand{\lpcp}{\text{LP}_\text{CP}}
\newcommand{\ipcp}{\text{BLP}_\text{CP}}
\newcommand{\minecc}{\textsc{MinECC}}
\newcommand{\ECC}{\textsc{ECC}}

\newcommand{\vc}{\textsc{Vertex Cover}}
\newcommand{\opt}[1]{\textsc{OPT}(#1)}

\newcommand{\LRVC}{\textsf{LocalRatioVC}}
\newcommand{\LRECC}{\textsf{LocalRatioECC}}

\usepackage{url}
\usepackage{hyperref}
\hypersetup{breaklinks, urlcolor=blue, colorlinks, citecolor=green!50!black, linkcolor=blue!80!black}

\title{An Improved Combinatorial Algorithm for Edge-Colored Clustering in Hypergraphs}

\author{Seongjune Han$^{1}$, Nate Veldt$^{2}$\\
	\\
	\normalsize{$^{1}$Department of Mathematics, Texas A\&M University, Texas, USA}\\
	\normalsize{$^{2}$Department of Computer Science and Engineering, Texas A\&M University, Texas, USA}\\
	\\
}
\date{}

\begin{document}

\maketitle

\begin{abstract}
Many complex systems and datasets are characterized by multiway interactions of different categories, and can be modeled as edge-colored hypergraphs. We focus on clustering such datasets using the NP-hard edge-colored clustering problem, where the goal is to assign colors to nodes in such a way that node colors tend to match edge colors. A key focus in prior work has been to develop approximation algorithms for the problem that are combinatorial and easier to scale. In this paper, we present the first combinatorial approximation algorithm with an approximation factor better than 2. 
\end{abstract}

\section{Introduction}
Graphs are a powerful model for representing pairwise relationships between entities in a dataset or complex real-world system. However, many systems and datasets involve multiway interactions among several entities at once (e.g., group emails, co-authored online articles, sets of co-reviewed retail products on e-commerce platforms) which cannot be fully captured by graphs. 
Furthermore, these multiway interactions are often associated with a certain {type} or {category} (e.g., email topic, venue where an article is published, IP address of product reviewer). One fundamental task for analyzing these types of datasets is to perform categorical clustering, where the goal is to identify groups of objects that tend to be in multiway interactions of a certain category. One formal objective for this problem is minimum edge-colored clustering in hypergraphs (\minecc{})~\cite{amburg2020clustering}. For this problem, the input is represented by an edge-colored hypergraph, where hyperedges represent multiway interactions and colors represent interaction types. The goal is to color nodes in a way that minimizes the weight of edges that are \emph{unsatisfied}. An edge is unsatisfied if it contains at least one node that is \emph{not} assigned the edge's color. Although the problem is NP-hard, many approximation algorithms have been developed. There has been a considerable focus on fast algorithms that can be made purely combinatorial (i.e., not relying on a global convex relaxation of the problem)~\cite{veldt2023optimal,lee2025improvedalgorithmsoverlappingrobust,amburg2020clustering,CraneEtAl2024Overlapping}. Faster algorithms for \minecc{} can serve as fundamental primitives for the various tasks that this problem has been motivated by and applied to, which include community detection~\cite{amburg2020clustering}, team formation~\cite{amburg2022Diverse,crane2025edge}, and resource allocation~\cite{angel2016clustering}. 

\textbf{Related work.}
The earliest work on edge-colored clustering focused on graphs, and includes work on approximation algorithms for maximizing the number of satisfied edges~\cite{angel2016clustering}, and fixed-parameter tractability results~\cite{cai2018alternatingpathcolouredclustering}. 
Amburg et al.~\cite{amburg2020clustering} extended the study to hypergraphs. They developed an approximation algorithm with approximation factor $\min\left(2 - \frac{1}{k}, 2 - \frac{1}{r+1}\right)$, by rounding a linear programming (LP) relaxation, where $k$ is the number of colors and $r$ is the maximum hyperedge size. Additionally, they introduced combinatorial algorithms with approximation factors that scale linearly with $r$. Amburg et al.~\cite{amburg2022Diverse} later built upon this work by applying the framework to problems like diverse group discovery and team formation, expanding the scope of $\minecc$ applications. Veldt~\cite{veldt2023optimal} developed a $\min\left(2 - \frac{2}{k}, 2 - \frac{2}{r+1}\right)$-approximation algorithm based on LP rounding, and combinatorial 2-approximations with runtimes that are linear in terms of the hypergraph size. More recently, Crane et al.~\cite{CraneEtAl2024Overlapping} generalized \ECC{} to allow for a budgeted number of overlapping cluster assignments or node deletions, and among other results presented greedy combinatorial $r$-approximation algorithms. Then, Lee et al.~\cite{lee2025improvedalgorithmsoverlappingrobust} used the primal-dual method to develop combinatorial methods for the overlapping objectives, achieving better solution quality than the greedy algorithms and faster performance than the LP-rounding approach. 
Although much of this prior work has focused on algorithms that are fast and combinatorial, there is no known combinatorial algorithm with an approximation ratio better than $2$. Furthermore, among linear-time algorithms, there is a deterministic 2-approximation for the unweighted case  and a randomized 2-approximation for the weighted case, but no deterministic 2-approximation for the weighted case~\cite{veldt2023optimal}.

\textbf{Our contributions.}
Our main contribution is a $(2-2/k)$-approximation for \minecc{} that can be made completely combinatorial and has a theoretical runtime that is nearly linear in terms of the size of the hypergraph. This matches the best theoretical approximation algorithm of Veldt~\cite{veldt2023optimal} in cases where $r \geq k-1$, and comes with a faster runtime both in theory and practice. The key to this contribution is to show how to quickly and combinatorially solve a certain type of LP relaxation for $\minecc$ by solving a minimum $s$-$t$ cut problem in an auxiliary graph. As a secondary contribution, we develop a deterministic 2-approximation algorithm that has a linear runtime and applies to weighted hypergraphs. We accomplish this by combining core ideas from other linear-time \minecc{} algorithms~\cite{veldt2023optimal} with a fast 2-approximation for weighted minimum vertex cover~\cite{BarYehuda1981}.

\section{Preliminaries}
\label{sec:prelim}
An instance of $\minecc$ is given by an edge-colored hypergraph $H = (V, E, C, \ell)$ where $V$ is the node set, $E$ is the edge set, $C = [k] = \{1, 2, 3, \cdots, k\}$ is the set of edge colors, and $\ell: E \rightarrow C$ is a map assigning an edge to a color in $C$. Let $r$ denote the maximum hyperedge size, and $\mu = \sum_{e \in E} |e|$ denote the size of the hypergraph. Let $E_i$ be the set of edges with color $i$. Let $E(u)$ be the set of edges incident to node $u \in V$.
For weighted instances, $w(e) \geq 0$ represents a nonnegative weight for $e \in E$. Define $w(A) = \sum_{e \in A} w(e)$. In $\minecc$, the goal is to color nodes so that the color matches the color of its containing edges as much as possible. Formally, if all nodes in a hyperedge $e \in E$ are given the color $\ell(e)$, then the hyperedge is \textit{satisfied}; otherwise it is \textit{unsatisfied}. \minecc{} seeks to color nodes to minimize the total weight of unsatisfied edges. When there is a node $v$ contained in the intersection of edges $e, f \in E$ with $\ell(e) \ne \ell(f)$, we cannot color $v$ in a way that satisfies both edges. 
We call such a pair $(e, f)$ a bad edge pair and let the set of all bad edge pairs be denoted as $\mathcal{B}$. Let $C(u) = \{j \in C \colon \exists e \ni u \text{ with } j = \ell(e)\}$ be the list of colors incident to node $u$. The canonical linear programming relaxation for $\minecc$ ($\lpecc$), is defined as follows:
\begin{equation}
\label{eq:canonical}
    \begin{array}{llll}
        \min & {\sum_{e\in E}} w(e) z_e & &\\ 
        \text{s.t.} & \forall e \in E & z_e \geq x_u^{\ell(e)} & \forall u \in e \\
                    & \forall u \in V & \sum_{i \in C} x_u^i = k-1 & \\
                    & \forall u \in V & 0 \leq x_u^i \leq 1 & \forall i \in C\\
                    & \forall e \in E & 0 \leq z_e \leq 1. &
    \end{array}
    \tag{$\lpecc$}
\end{equation}
When variables $x_u^i, z_e$ are binary, the relaxation becomes a binary linear program that exactly captures the \minecc{} objective. The variable $x_u^i$ for a node $u$ and a color $i$ can be interpreted as a distance. If variables are binary, $x_u^i = 0$ indicates that $u$ is given color $i$, and the variable $z_e$ captures whether an edge $e$ is satisfied $(z_e = 0)$ or unsatisfied $(z_e = 1)$. We denote the optimal solution value for $\lpecc$ as $\opt{\lpecc}$, and use similar notation for other LPs. 

\textbf{Vertex-cover reduction.} 
Veldt~\cite{veldt2023optimal} showed $\minecc$ can be reduced to $\vc$ in an approximation preserving way as a means to developing fast $2$-approximations. To see this reduction it is helpful to see that \minecc{} can equivalently be viewed as an edge deletion task, where the goal is to delete a minimum weight set of edges (or number of edges in the unweighted case) so that all nodes belong to edges of at most one color in the remaining hypergraph. This in turn is equivalent to deleting a minimum weight set of edges so that no bad edge pairs remain, then we can color nodes according to the remaining edges. This is analogous to \vc{}, the task of identifying a minimum weight set of nodes to cover all edges.
In more detail, when reducing \minecc{} to $\vc$, each edge $e \in E$ in the hypergraph is associated with a node $v_e$ in a $\vc$ instance, and two nodes $(v_e, v_f)$ define an edge if $(e,f) \in \mathcal{B}$ (see Figure~\ref{fig:allfigures}). 
This reduction leads to an alternative type of LP relaxation $\lpvc$ for $\minecc$:
\begin{equation}
\label{eq:lpvc}
\begin{array}{lll}
\min  & \sum_{e \in E} w(e)x_e &\\
\text{s.t.} & x_e + x_f \geq 1 & \forall (e,f) \in \mathcal{B}\\
                &x_e \geq 0  & \forall e \in E. 
\end{array}
\tag{$\lpvc$}
\end{equation}
This LP amounts to reducing $\minecc$ to $\vc$ and then taking the standard LP relaxation for $\vc$.
The constraint $x_e + x_f \geq 1$ captures the fact that at least one edge in every bad edge pair must be deleted. A standard fact from prior work on $\vc$ is that this LP always has an optimal solution that is half-integral, meaning $x_e \in \{0,1, 1/2\}$ for every $e \in E$~\cite{nemhauser1975vertex}.

\section{A Deterministic 2-approximation Algorithm}
\label{sec:deterministic}
Veldt~\cite{veldt2023optimal} presented a randomized 2-approximation algorithm that runs in $O(\mu)$ time and applies even to weighted hypergraphs. This algorithm is obtained by implicitly applying the randomized 2-approximation for \vc{} by Pitt~\cite{pitt1985simple} without forming a reduced graph. In order to obtain a linear-time 2-approximation for \minecc{} that applies to weighted \minecc{} and is deterministic, we follow a similar approach but instead implicitly implement an existing deterministic algorithm for \vc{}. In particular, Algorithm~\ref{alg:localratiovc} is the deterministic 2-approximation algorithm for node-weighted \vc{} of Bar-Yehuda and Even~\cite{BarYehuda1981, BarYehuda1985}. It follows a greedy strategy that iterates through edges and locally lowers the weights of two adjacent nodes in the \vc{} instance by the minimum of their current weights. After iterating through all edges, the algorithm selects every node whose weight becomes zero. This algorithm can be implemented in linear time in terms of the number of edges in the \vc{} instance. However, if we explicitly reduce \minecc{} to a \vc{} instance and then apply Algorithm~\ref{alg:localratiovc}, this would take superlinear time in terms of the hypergraph size.

In Algorithm~\ref{alg:localratioecc}, we provide pseudocode for our algorithm \LRECC{}, which implicitly applies Algorithm~\ref{alg:localratiovc} to the \minecc{} instance $H = (V, E, C, \ell)$ without forming the reduced \vc{} instance. For this algorithm, we assume that the edges $E$ are sorted by color in increasing order. For each node $v \in V$, we examine its incident edges $E(v)$, where $v(1), v(2), \ldots, v(d_v)$ denote the indices of these edges and $d_v$ is the degree of $v$. Algorithm~\ref{alg:localratiovc} iterates through edges of the \vc{} instance, which would correspond to iterating through bad edge pairs in the \minecc{} instance. To iterate through all bad edge pairs, we visit each node $v \in V$ and consider pairs of edges that are in $E(v)$. 
We use two pointers, $f$ and $b$, initialized to $1$ and $d_v$, respectively. At each step, either $f$ is incremented or $b$ is decremented. When the pair $(e_{v(f)}, e_{v(b)})$ forms a bad edge pair, we reduce the weights of both edges by the minimum of their current weights, following the local-ratio update in Algorithm~\ref{alg:localratiovc}. As a result, at least one of the two edges ends up with weight zero. Since the algorithm eventually collects all edges whose weights become zero, which corresponds exactly to the set of edges to be deleted, we can immediately delete the edge whose weight becomes zero. This approach avoids explicitly iterating over all bad edge pairs by skipping pairs involving edges that have already been deleted. This guarantees a total running time of $O(\mu)$. Algorithm~\ref{alg:localratioecc} follows the same bad edge pair traversal strategy as Veldt~\cite{veldt2023optimal}. The difference is that when visiting a bad edge pair $(e_{v(f)}, e_{v(b)})$. it deterministically deletes at least one hyperedge by implicitly implementing \LRVC{}, rather than randomly deleting a hyperedge as in Pitt’s algorithm. We summarize with the following theorem.
\begin{theorem}
    Algorithm~\ref{alg:localratioecc} is a 2-approximation algorithm for \minecc{}. It runs in $O(\sum_{v \in V}d_v) = O(\mu)$ time.
\end{theorem}
\begin{proof}
    Since Algorithm~\ref{alg:localratioecc} is an implicit implementation of \LRVC{} and the reduction from \minecc{} to \vc{} is an approximation preserving reduction, Algorithm~\ref{alg:localratioecc} is a 2-approximation. When it visits bad edge pairs sharing node $v$, it performs $O(d_v)$ work, since at least one hyperedge containing $v$ is deleted at each step. Therefore it runs in $O(\sum_{v \in V} d_v) = O(\mu)$ time overall. 
\end{proof}

\begin{algorithm}[t]
\caption{\LRVC{}}
\label{alg:localratiovc}
\begin{algorithmic}
    \State \textbf{Input:} A graph $G = (V, E)$
    \State \textbf{Output:} A vertex cover $C$
    \State $C \gets \emptyset$
    \For{$v \in V$} 
    \State $r(v) = w(v)$
    \EndFor
    \For{$(u,v) \in E$} 
    \State $M \gets \min\{r(u), r(v)\}$
    
    \State $r(u) \gets r(u) - M$
    \State $r(v) \gets r(v) - M$
    \EndFor
    \State $C \gets \{v \in V \mid r(v) = 0\}$
\end{algorithmic}
\end{algorithm}

\begin{algorithm}[H]
\caption{\LRECC{}}
\label{alg:localratioecc}
\begin{algorithmic}
 \State \textbf{Input:} A weighted instance $H = (V, E, C, \ell)$ 
\State \textbf{Output:} $D \subseteq E$ of edges to delete
\State $D \gets \emptyset$
\For{$v \in V$}
    \State $E(v) = [e_{v(1)} \hspace{7pt} e_{v(2)} \hspace{7pt} \cdots \hspace{7pt} e_{v(d_v)}]$
    \State $f=1, b=d_v$
    \While{$e_{v(b)} \in D$ and $b>f$}
        \State $b \gets b-1$
    \EndWhile

    \While{$e_{v(f)} \in D$ and $b>f$}
        \State $f \gets f+1$
    \EndWhile

    \While{$\ell(e_{v(f)}) \neq \ell(e_{v(b)})$}
        \If{$w(e_{v(f)}) < w(e_{v(b)})$}
            \State $D \gets D \cup \{e_{v(f)}\}$
            \State $w(e_{v(b)}) \gets w(e_{v(b)}) - w(e_{v(f)})$
            \While{$e_{v(f)} \in D$ and $b>f$}
                \State $f \gets f+1$
            \EndWhile
        \ElsIf{$w(e_{v(f)}) = w(e_{v(b)})$}
            \State $D \gets D \cup \{e_{v(f)}, e_{v(b)}\}$
            \While{$e_{v(f)} \in D$ and $b>f$}
                \State $f \gets f+1$
            \EndWhile
            \While{$e_{v(b)} \in D$ and $b>f$}
                \State $b \gets b-1$
            \EndWhile
        \Else
            \State $D \gets D \cup \{e_{v(b)}\}$
            \State $w(e_{v(f)}) \gets w(e_{v(f)}) - w(e_{v(b)})$
            \While{$e_{v(b)} \in D$ and $b>f$}
                \State $b \gets b-1$
            \EndWhile
        \EndIf
    \EndWhile
\EndFor
\State Return $D$
\end{algorithmic}
\end{algorithm}

\section{A Faster $(2-2/k)$-Approximation }
\label{sec:2-2/kalgorithm}
This section details our new $(2-2/k)$-approximation for \minecc{} that can be made completely combinatorial. The approximation guarantee will be obtained by reducing to an instance of $k$-colorable $\vc$ and then applying an approximation algorithm for that problem~\cite{hochbaum1983efficient}. This amounts to applying a rounding algorithm for $\lpvc$. The more important contribution is a new and more efficient approach to solve $\lpvc$ combinatorially without actually forming the $\vc$ instance or setting up the LP explicitly.

\subsection{The LP rounding algorithm}
\label{subsec:rounding}
The reduction from $\minecc$ to $\vc$ generates a naturally $k$-colorable graph, by associating hyperedge colors with colors for nodes in the reduced $\vc$ instance. This is because by design, hyperedges with the same color cannot define bad edge pairs. This observation is beneficial since Hochbaum~\cite{hochbaum1983efficient} previously showed how to round the LP relaxation for $\vc$ into a $(2 - 2/k)$-approximate solution for $k$-colorable graphs.
Instead of explicitly forming a $\vc$ instance when reducing from \minecc{}, we can apply the approach of Hochbaum (which amounts to Algorithm~\ref{alg:rounding}) to directly round a half-integral solution to $\lpvc$.
\begin{algorithm}[t]
\caption{Approximate LP rounding algorithm}
\label{alg:rounding}
\begin{algorithmic}
    \State \textbf{Input:} An optimal half-integral solution $X = \{x_e \mid e \in E\}$ to $\lpvc$ 
    \State \textbf{Output:} $Y \subseteq E$ of edges to delete
    \State $E^1 = \{e \in E \colon x_e = 1 \}$
    \State $E^{1/2} = \{e \in E \colon x_e = 1/2 \}$
    \State $E^0 = \{e \in E \colon x_e = 0 \}$
    \State $i^* \gets \text{argmax}_{i\in C}(w(E_i \cap E^{1/2}))$
    \State \Return $Y:= E^1 \cup (E^{1/2} - E_{i^*}\cap E^{1/2})$\;
\end{algorithmic}
\end{algorithm}
\begin{theorem}
    Algorithm~\ref{alg:rounding} is a $(2-\frac2k)$-approximation for \minecc{}.
\end{theorem}
\begin{proof}
Given a hypergraph $H = (V, E, C, \ell)$ and an optimal half-integral solution $\{x_e^* \mid e\in E\}$ to $\lpvc$, first, let us show that $\opt{\lpvc} := \sum_{e \in E} w(e)x_e^* = w(E^1) + \frac{1}{2} w(E^{1/2})$.

\begin{align*}
    \sum_{e \in E} w(e)x_e^* &= \sum_{e \in E^1} w(e)x_e^* + \sum_{e \in E^{1/2}} w(e)x_e^* + \sum_{e \in E^0} w(e)x_e^* \\
    &= \sum_{e \in E^1} w(e) \cdot 1 + \sum_{e \in E^{1/2}} w(e) \cdot \frac{1}{2} + \sum_{e \in E^0} w(e) \cdot 0\\
    &= w(E^1) + \frac{1}{2} w(E^{1/2}).
\end{align*}
\\   
Second, from the definition of $i^*$ we know that
\[
w(E_{i^*} \cap E^{1/2}) \geq \frac{1}{k} w(E^{1/2}).
\]
Next, note that if $(e,f)$ is a bad edge pair where $e \in E^0$, then $f$ must be in $E^1$ in order to satisfy $x_e + x_f \geq 1$. Therefore, $E^0 \cup (E_{i^*} \cap E^{1/2})$ is a set of satisfiable edges. Now, consider its complement $E^1 \cup (E^{1/2} - E_{i^*} \cap E^{1/2})$. Then, its weight (assuming $k \geq 2$) is
\begin{align*}
    w(E^1 \cup (E^{1/2} - E_{i^*} \cap E^{1/2})) &= w(E^1) + w(E^{1/2}) - w(E_{i^*} \cap E^{1/2})\\
    &\leq w(E^1) + w(E^{1/2}) - \frac{1}{k} w(E^{1/2})\\
    &\leq \left(2 - \frac{2}{k}\right) \left(w(E^1) + \frac{1}{2} w(E^{1/2})\right) \\
    &= \left(2 - \frac{2}{k}\right) \opt{\lpvc}\\ 
    &\leq \left(2 - \frac{2}{k}\right) \opt{\minecc}, 
\end{align*}
where $\opt{\minecc}$ is the sum of weights for an optimal solution to $\minecc$ for $H$.
\end{proof}

Hochbaum~\cite{hochbaum1983efficient} showed that the $\vc$ LP relaxation can be solved by reducing it to a minimum $s$-$t$ cut in an auxiliary graph, where the number of edges in the auxiliary graph is roughly twice the number of edges in the \vc{} instance. Applying this approach to a \minecc{} instance $H = (V,E,C, \ell)$ would result in an $s$-$t$ cut problem in an auxiliary graph with $2|E|+2$ nodes and $2|E| + 2|\mathcal{B}| = O(|E|^2)$ edges. The resulting minimum $s$-$t$ cut problem can then in theory be solved in $\tilde{O}(|E|^2)$-time using the algorithm of van den Brand et al.~\cite{van2021minimum}. This is already faster than the best approximation obtained by rounding $\lpecc$~\cite{veldt2023optimal}, which gives the same $(2-2/k)$-approximation when $r \geq k-1$. The best theoretical runtime for solving an LP of the form $\min_{Ax=b;x \geq 0}c^Tx$ with $N$ variables is $\Omega(N^\omega)$-time~\cite{Jiang2021Faster, cohen2020solvinglinearprogramscurrent}. When we convert $\lpecc$ to this form, the number of variables is $N = |E| + k|V| + \mu$. Since $\omega \geq 2$, this results in a runtime of $\Omega((k|V| + \mu)^2)$. 

\textbf{Key challenge.} Although faster in theory, reducing explicitly to \vc{} and applying the results of Hochbaum~\cite{hochbaum1983efficient}
is impractical because the reduced \vc{} instance becomes very dense. For many real \minecc{} datasets (see Table~\ref{tab_experiments}), $|\mathcal{B}|$ is extremely large. Even \emph{finding} $\mathcal{B}$ explicitly is expensive, and takes superlinear time in terms of the hypergraph size even if $r$ and $k$ are constants.
Meanwhile, although the theoretical runtime for $\lpecc$ is bad, it often has far fewer constraints than $\lpvc$ in practice (see Table~\ref{tab_experiments}), and it has been shown to be easily solved in practice for many real-world instances~\cite{amburg2020clustering}.
Our main contribution is a new way to solve $\lpvc$ that completely avoids dealing with bad edge pairs, resulting in better runtimes both in theory and practice.

\subsection{Color Pair Binary Linear Program}
To avoid the challenge above, we present a binary linear program that is equivalent to $\lpvc$ but has far fewer constraints and can be reduced to a much sparser minimum $s$-$t$ cut problem. The idea behind this binary linear program comes from blending concepts from $\lpecc$ and $\lpvc$. We first define a new linear program---the color-pair LP ($\lpcp$)---that merges constraints from $\lpecc$ and $\lpvc$:
\begin{equation}\label{eq:lpcp}
    \begin{array}{lll}
        \min & {\displaystyle \sum_{e\in E}} w(e) z_e & \\ 
        \text{s.t.} & \forall e \in E, \quad z_e \geq x_u^{\ell(e)} & \forall u \in e\\
        &\forall u \in V, \quad x_u^i + x_u^j \geq 1 & \forall (i, j) \in {C(u) \choose 2}\\[3pt]
        & \forall u \in V, \quad 0 \leq x_u^i \leq 1 & \forall i \in C(u)\\
        & \forall e \in E, \quad 0 \leq z_e \leq 1.&
    \end{array}
\tag{2}
\end{equation}
In this relaxation, the constraint $(x_u^i + x_u^j \geq 1)$ captures the fact that a node $u$ cannot have both color $i$ and color $j$. This constraint resembles the bad edge pair constraint $x_e + x_f \geq 1$ from $\lpvc$, but we will not need $\Omega(|\mathcal{B}|)$ constraints of this form. This constraint relaxes the constraint $\sum_{c \in C} x_u^c = k-1$ from $\lpecc$. This leads to a looser relaxation. $\lpcp$ can be shown to be equivalent to the following color-pair binary linear program ($\ipcp$):
\begin{equation*}
    \begin{array}{lll}
        \min &\displaystyle\frac{1}{2} \sum_{e\in E}  w(e) (b_e - a_e + 1) \\ 
        \text{s.t.} & \forall e \in E: \quad  a_e \leq a_u^{\ell(e)} & \forall u \in e\\
        & \forall e \in E: \quad b_u^{\ell(e)} \leq b_e  & \forall u \in e\\
        &\forall u \in V: \quad a_u^i \leq b_u^j & \forall (i, j) \in {C(u) \choose 2}\\[3pt]
        &\forall u \in V: \quad a_u^j \leq b_u^i & \forall (i, j) \in {C(u) \choose 2}\\
        &\forall u \in V: \quad a_u^i, b_u^i \in \{0, 1\} & \forall i \in C(u)\\
        &\forall e \in E: \quad a_e, b_e \in \{0, 1\}.
    \end{array}
\tag{$\ipcp$}
\end{equation*}
We prove that this BLP is equivalent to $\lpvc$, which is perhaps surprising given that it does not include any explicit reference to bad edge pairs. We remark that it is not strictly necessary to consider or prove anything about $\lpcp$ in order to see that $\ipcp$ is equivalent to $\lpvc$. However, it serves as a helpful conceptual stepping stone. In particular, it is relatively easy to see that $\lpcp$ is a slight relaxation of $\lpecc$. It differs in that it always permits a half-integral solution, and is equivalent with $\ipcp$ and $\lpvc$. Rather than proving all of these equivalences, it suffices for our analysis to prove directly that $\ipcp$ is equivalent to $\lpvc$ and can be solved via reduction to a minimum $s$-$t$ cut problem.

\begin{theorem}
\label{thm:ipcplpvc}
Let $\{a_e, b_e,a_u^i, b_u^i \mid e\in E, u \in V, i \in C\}$ be a set of optimal variables for $\ipcp$ and for $e \in E$ define:
\begin{equation}
\label{eq:xe}
    x_e = \frac{1}{2}(b_e -a_e + 1).
\end{equation}
Then, $\{x_e\}$ is an optimal set of variables for $\lpvc$. 
\end{theorem}
\begin{proof}
We first show that $\{x_e\}$ variables as defined in Eq.~\eqref{eq:xe} are feasible for $\lpvc$.
If $(e, f)$ is a bad edge pair, there exists some $u\in e \cap f$ such that the following inequalities hold:
\begin{align*}
    a_e \leq a_u^{\ell(e)} \leq b_u^{\ell(f)} \leq b_f \quad \text{and} \quad a_f \leq a_u^{\ell(f)} \leq b_u^{\ell(e)} \leq b_e.
\end{align*}
This implies $x_e + x_f = \frac{1}{2}(b_e+b_f-a_e-a_f+2) \geq 1$. Hence, we have a feasible solution for $\lpvc$ with the same objective value, so $\opt{\ipcp} \geq \opt{\lpvc}$. 

Now, we show $\opt{\ipcp} \leq \opt{\lpvc}$. Let $\{x_e\}$ be an optimal solution to $\lpvc$, and define $a_e, b_e$ as follows:
\[
(a_e, b_e) := 
\begin{cases}
(1, 0) & \text{if } x_e = 0\\
(0, 1) & \text{if } x_e = 1\\
(1, 1) & \text{if } x_e = 1/2.
\end{cases}
\]
Notice that $x_e = (b_e - a_e + 1)/2$, meaning the objective functions of $\lpvc$ and $\ipcp$ are equal. We also define $a_u^i$ and $b_u^i$ as follows:
\begin{equation*}
     a_u^i := \max_{e \in E_i, u \in e} a_e  \quad \text{and} \quad  b_u^i := \min_{e \in E_i, u \in e} b_e.
\end{equation*}
This implies $a_e \leq a_u^{\ell(e)}$ and $b_u^{\ell(e)} \leq b_e$ for all $u \in e$. Now, suppose $a_u^i > b_u^j$ for $i \neq j$. Then, $a_u^i = 1$ and $b_u^j = 0$. This tells us that there exists some $e \in E_i$ such that $u \in e$ and $a_e = 1$. By the definition of $a_e$, having $a_e = 1$ means $x_e \leq \frac{1}{2}$. Similarly, there exists some $f \in E_j$ such that $u \in f$ and $b_f = 0$. By the definition of $b_e$, having $b_f = 0$ means $x_f = 0$. Since $u \in e \cap f$ and $i \ne j$, $(e, f)$ is a bad edge pair, but $x_e + x_f \leq \frac{1}{2}$ which contradicts with the constraint $x_e + x_f \geq 1$ in $\lpvc$. Thus, $a_u^i \leq b_u^j$ is true. Using an analogous argument, we can also prove that the constraint $a_u^j \leq b_u^i$ also holds.
\end{proof}
In $\ipcp$, the number of constraints is $2\sum_{u \in V} |C(u)| + 2\mu + 2|E| + 2\sum_{u \in V}{|C(u)| \choose 2} =O(\mu + k^2|V|)$. 
This does not depend on $|\mathcal{B}|$ and enables us to avoid the issue discussed in Section~\ref{subsec:rounding}. If $k$ is constant, this value is linear in terms of the hypergraph size $\mu$.
\begin{figure}
    \centering
    \includegraphics[width = .8\linewidth]{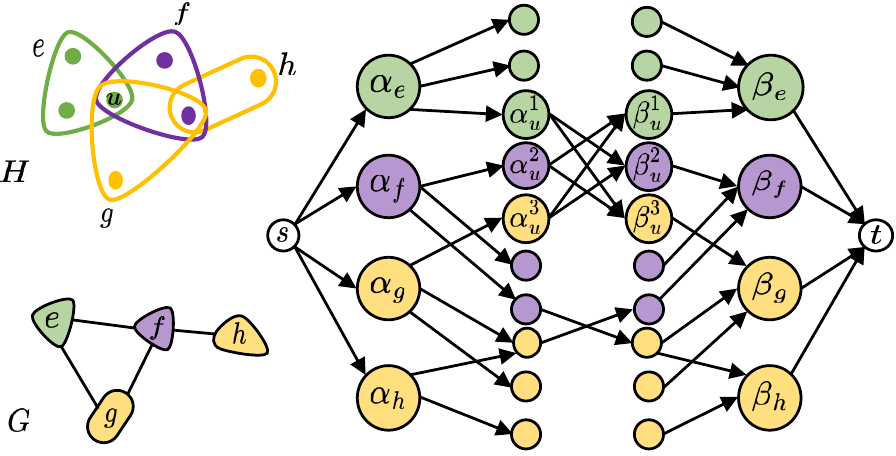}
    \caption{In the top left we show a \minecc{} instance ${H}$, which can be reduced to the $k$-colorable \vc{} instance $G$ in the bottom left. On the right, we show the reduced flow network whose minimum $s$-$t$ cut coincides with the optimal solution to $\ipcp$ (see Section~\ref{subsec:reductionalgorithm}).}
    \label{fig:allfigures}
\end{figure}

\subsection{Reduction to Flow Network}
\label{subsec:reductionalgorithm}
In addition to having far fewer constraints than $\lpvc$, $\ipcp$ can be reduced to a new minimum $s$-$t$ cut problem with fewer nodes and edges than the graph obtained by explicitly reducing to \vc{} and applying the results of Hochbaum~\cite{hochbaum1983efficient}. 
Formally, we use the following flow network construction for an instance $H = (V,E,C,\ell)$:
\begin{enumerate}[leftmargin = 20pt]
    \item Introduce source and sink nodes $s$ and $t$, and initialize empty node sets $A_E, B_E, A_V,$ and $B_V$.
    \item For $e \in E$, add nodes $\alpha_e \in A_E$ and $\beta_e \in B_E$.
    \item For $u \in V$ and $i \in C(u)$, add nodes $\alpha_u^i \in A_V$ and $\beta_u^i \in B_V$.
    \item For $u \in V$ and $(i, j) \in {C(u) \choose 2}$, add directed edges $(\alpha_u^i, \beta_u^j)$ and $(\alpha_u^j, \beta_u^i)$ with weight $\infty$.
    \item For $e \in E$ and each $u \in e$, introduce directed edges $(\alpha_e, \alpha_u^{\ell(e)})$ and $(\beta_u^{\ell(e)}, \beta_e)$ with weight $\infty$.
    \item For each $e \in E$, introduce directed edges $(s,\alpha_e)$ and $(\beta_e,t)$ with weight $\frac{1}{2} w(e)$.
\end{enumerate}

\newcommand{\detours}{\LRECC{}}
\newcommand{\ours}{\textsf{ColorPair-Flow}}
\newcommand{\hoch}{\textsf{VC-Flow}}
\newcommand{\old}{$\lpecc$-\textsf{round}}

\begin{table*}[t]
\caption{Statistics for six benchmark \minecc{} datasets.}
\label{tab_experiments}
\begin{center}
\scalebox{0.85}{
    \begin{tabular}{ccccccccccc}
    \toprule
    & \multicolumn{7}{c}{\textbf{Hypergraph Statistics}}
    & \multicolumn{3}{c}{\textbf{Constraints}}\\ 
    \cmidrule(lr){2-8}
    \cmidrule(lr){9-11}\
     \emph{Dataset}  & {\small	$|V|$} & {\small	$|E|$} & {\small	$k$} & {\small	$r$} & {\small	$\mu$} & {\small $|\mathcal{B}|$} & {\small	$(2-2/k)$} & {\small	$\lpecc$} & {\small	$\lpvc$} & {\small	$\lpcp$}\\ 
    \midrule
    {\small 	\texttt{Brain}} & 638 & 21,180 & 2 & 2 & 42,360 & 490,163 & 1 & 42,998 & 490,163 & 42,998\\
    {\small 	\texttt{MAG-10}} & 80,198 & 51,889 & 10 & 25 & 180,726 & 473,780 & 1.8 & 260,924 & 473,780 & 910,795\\
    {\small 	\texttt{Cooking}} & 6,714 & 39,774 & 20 & 65 & 428,275 & 312,098,857 & 1.9 & 434,989 & 312,098,857 & 657,292 \\
    {\small 	\texttt{DAWN}} & 2,109 & 87,104 & 10 & 22 & 343,211 & 340,602,559 & 1.8 & 345,320 & 340,602,559 & 374,775 \\
    {\small 	\texttt{Walmart}}  &   88,837 & 65,898 & 44 & 25 & 452,208 & 27,283,314 & 1.955 & 541,045 & 27,283,314 & 4,680,597 \\
    {\small 	\texttt{Trivago}} &    207,974 & 247,362 & 55 & 85 & 757,946 & 2,449,951 & 1.964 & 965,920 & 2,449,951 & 12,210,823 \\
    \bottomrule
    \end{tabular}
}
\end{center}
\end{table*} 

\begin{table*}[t]
\caption{Comparisons for three algorithms. The ratio columns indicate the objective score divided by the LP lower bound for each algorithm. We display the average and standard deviation over 5 runs.}
\label{tab_experiments2}
\begin{center}
\scalebox{0.85}{
    \begin{tabular}{cccccccccc}
    \toprule
    & \multicolumn{3}{c}{\old{}}
    & \multicolumn{3}{c}{\hoch{}}
    & \multicolumn{3}{c}{\ours}\\ 
    \cmidrule(lr){2-4} 
    \cmidrule(lr){5-7} 
    \cmidrule(lr){8-10} \
     \emph{Dataset}  & {\small    ratio} & {\small    runtime (s)} & {\small    memory} & {\small  ratio} & {\small  runtime (s)} & {\small    memory} & {\small  ratio} & {\small  runtime (s)} & {\small    memory}\\ 
    \midrule
    {\small 	\texttt{Brain}} & 1 & 0.191$\pm$0.009 & 54MB & 1.002 & 5.443$\pm$0.196 & 1.1GB & 1 & 0.526$\pm$0.179 & 209MB\\
    {\small 	\texttt{MAG-10}} & 1 & 5.383$\pm$0.378 & 861MB & 1.193 & 12.413$\pm$0.163 & 1.2GB & 1.193 & 4.454$\pm$0.142 & 821MB \\
    {\small 	\texttt{Cooking}} & 1 & 18.465$\pm$1.647 & 540MB & N/A & timed out & - & 1.605 & 1.599$\pm$0.038 & 1.1GB \\
    {\small 	\texttt{DAWN}} & 1 & 2.888$\pm$0.121 & 398MB & N/A & timed out & - & 1 & 5.582$\pm$0.177 & 1.2GB \\
    {\small 	\texttt{Walmart}} & 1 & 631.559$\pm$4.443 & 3.5GB & N/A & timed out & - & 1.654 & 6.794$\pm$0.042 & 2.4GB\\
    {\small 	\texttt{Trivago}} & 1.001 & 51.124$\pm$7.399 & 9.5GB & 1.272 & 299.06$\pm$2.25 & 6.1GB & 1.272 & 54.646$\pm$0.159 & 4.1GB\\
    \bottomrule
    \end{tabular}
}
\end{center}
\end{table*} 

After forming this graph, we solve the maximum $s$-$t$ flow problem over it, and use it to compute a minimum $s$-$t$ cut using standard techniques. This partitions nodes into two sets, $S$ and $\bar{S}$. Each node in the graph corresponds to one binary variable in $\ipcp$: node $\alpha_e$ corresponds to variable $a_e$, node $\beta_e$ corresponds to $b_e$, and similarly we have the mapping $\alpha_u^i \leftrightarrow a_u^i$, and $\beta_u^i \leftrightarrow b_u^i$.
A variable is set to 1 if its corresponding node is in $S$, otherwise the variable is 0.

Now, we explain how the construction exactly captures constraints in $\ipcp$. If there is an $\infty$-weight directed edge from $\alpha$ to $\beta$, then $\beta \in \bar{S} \implies \alpha \in \bar{S}$, and $\alpha \in S \implies \beta \in S$, to avoid cutting and $\infty$-weight edge. If $\alpha$ and $\beta$ are associated with variables $a$ and $b$, this is equivalent to a constraint of the form $a \leq b$, which enforces that $a = 1 \implies b = 1$ and $b = 0 \implies a = 0$.

Observe finally that the minimum cut value is exactly equal to $\opt{\ipcp}$. Consider an edge $e \in E$. If $\alpha_e \in S, \beta_e \in \bar{S}$, then neither $(s, \alpha_e)$ nor $(\beta_e, t)$ is cut.
If $\alpha_e \in \bar{S}, \beta_e \in S$, then both $(s, \alpha_e)$ and $(\beta_e, t)$ are cut. This contributes $2 \times \frac{1}{2}w(e)=w(e)$ to the cut value. If $\alpha_e, \beta_e \in \bar{S}$ or $\alpha_e, \beta_e \in S$, then either $(s, \alpha_e)$ or $(\beta_e, t)$ is cut, respectively. In either case, this contributes $\frac{1}{2}w(e)$ to the cut. We obtain the same contributions to the objective of $\ipcp$ when we consider what the above cases imply about the variables $a_e$ and $b_e$.

\textbf{Improved runtime.}
We now analyze the runtime of solving the new minimum $s$-$t$ cut problem. The number of nodes in the reduced flow network is
$N = 2 + 2|E| + 2 \sum_{u \in V} |C(u)| \leq 2 + 2|E| + 2k|V|$
and the number of edges is
\[
M = 2|E| + 2\sum_{e \in E} |e| + 2 \sum_{u \in V} {|C(u)| \choose 2} \leq 2|E| + 2\mu + (k^2 - k)|V|.
\] 
Given a directed graph with $N$ nodes and $M$ edges, the algorithm of van den Brand et al.~\cite{van2021minimum} finds a minimum $s$-$t$ cut in $\tilde{O}(M+N^{1.5})$ time. Chen et al.~\cite{chen2022maximum} solves the problem in $M^{1+o(1)}$ time. From this we know that $\ipcp$ can be solved in $(\mu + k^2|V|)^{1+o(1)}$ time or { $\tilde{O}(\mu + k^2|V| + (|E| + k|V|)^{1.5})$} time. This is significantly faster than the runtime obtained by reducing explicitly to \vc{} and applying the approach of Hochbaum~\cite{hochbaum1983efficient} (see Section~\ref{subsec:rounding}). In particular, if $k$ is a constant, our best theoretical runtime is nearly linear in terms of the hypergraph size.

\section{Experimental Results}
\label{sec:Implementation}
We call our new $(2-2/k)$-approximation for \minecc{} \ours{}, since it amounts to solving a maximum $s$-$t$ flow problem derived from the color-pair BLP. This algorithm uses Theorem~\ref{thm:ipcplpvc} to obtain a solution to $\lpvc$, and then uses Algorithm~\ref{alg:rounding} to delete certain edges so that nodes can be colored in a way that leaves all remaining edges satisfied. We implement this in Julia and run experiments on a standard suite of benchmark \minecc{} datasets. See Table~\ref{tab_experiments} for datasets statistics and prior work for more information about these hypergraphs~\cite{amburg2020clustering,veldt2023optimal}. We run experiments on a research server with 1TB of RAM. Code is available at~\url{https://github.com/shan-mcs/minecc-combinatorial/}. We use a Julia implementation of the push-relabel algorithm for solving maximum flows, originally used in the work of Veldt et al.~\cite{veldt2019flow} and improved by Huang et al.~\cite{huang2024densest}. The results in Table~\ref{tab_experiments2} show that this method typically finds solutions that are much better than the worst case guarantee of $(2-2/k)$, and sometimes even finds optimal solutions.

For comparison we implement the minimum $s$-$t$ cut construction that follows from the work of Hochbaum~\cite{hochbaum1983efficient} for \vc{} (\hoch; see Section~\ref{sec:2-2/kalgorithm}). This approach faces significant memory and runtime issues, since it requires explicit enumeration and storage of all bad edge pairs $\mathcal{B}$, which can be extremely large. Not surprisingly, this method typically runs out of time (we set an upper bound of 15 minutes). We also attempted to directly solve $\lpvc$ using Gurobi optimization software, but this ran into memory and runtime issues for the same reason. This comparison highlights the benefits of our approach. $\lpvc$ is a useful LP relaxation to consider, as it has half-integral solutions and is particularly amenable to combinatorial approaches. However, direct approaches for solving it are completely infeasible since $|\mathcal{B}|$ can be huge even for medium sized hypergraphs. We stress that the poor performance of this approach for \minecc{} is not due in any way to limitations in the work of Hochbaum~\cite{hochbaum1983efficient}, which focuses on \vc{}. Rather, the poor performance is due to issues in explicitly reducing \minecc{} to \vc{}. Our approach sidesteps this issue and solves the LP without ever enumerating $\mathcal{B}$. 

We also solve $\lpecc$ using Gurobi and round it via a standard approach: node $u$ is given color $c = \arg\max_i x_u^i$~\cite{amburg2020clustering,veldt2023optimal}. This produces better quality solutions (since $\lpecc$ provides a tighter lower bound), and in some cases is even slightly faster than \ours. However, this approach can have unusually long runtimes. 
On Walmart, \old{} takes over 10 minutes while \ours{} takes under 10 seconds. It is not immediately clear which hypergraph characteristics lead to particularly long runtimes. In particular, Walmart does not have the highest value for $|V|$, $|E|$, $r$, $k$, $\mu$, or $|\mathcal{B}|$. In contrast, \ours{} provides a favorable runtime and quality tradeoff, as it produces good solutions with much more predictable and consistently fast runtimes. We also tracked the memory allocated by each method. \ours{} required comparatively less memory on the large datasets. For Walmart and Trivago, \ours{} allocated 2.4GB and 4.1GB respectively, whereas the $\lpecc$-based approach allocated 3.5GB and 9.5GB.

\section{Conclusions}
We have presented a combinatorial $(2-2/k)$-approximation algorithm for \minecc{} that comes with a much faster theoretical runtime than solving and rounding a canonical LP relaxation. In practice, the canonical LP still produces extremely good results and is often fast, but in some cases can also be very slow. An open direction is to try to obtain even better empirical runtimes for our \ours{} by using even faster max-flow solvers.

\section*{Acknowledgments}
Both authors are supported by AFOSR Award Number FA9550-25-1-0151. Nate Veldt is additionally supported by ARO Award Number W911NF24-1-0156. The views and conclusions contained in this document are those of the authors and should not
be interpreted as representing the official policies, either expressed or implied, of the Army Research Office
or the U.S. Government.

\bibliographystyle{plain}
\bibliography{refs}

\end{document}